\documentclass[a4paper,UKenglish,cleveref, autoref, thm-restate]{lipics-v2021}

\hideLIPIcs  

\graphicspath{{./figures/}}
\usepackage{booktabs}
\usepackage{tikz}
\title{Belief-Aware Pivotal Mechanism for DAO Committees} 

\author{Nuno Braz}{INESC-ID, Instituto Superior Técnico, Universidade de Lisboa\and \url{https://nunobrazz.github.io/} }{}{https://orcid.org/0009-0003-1151-3696}{}

\author{Diogo Poças}{Instituto de Telecomunicações, Instituto Superior Técnico, Universidade de Lisboa\and \url{https://diogopocas1991.gitlab.io/} }{}{https://orcid.org/0000-0001-5193-2260}{}

\author{Miguel Correia}{INESC-ID, Instituto Superior Técnico, Universidade de Lisboa}{miguel.p.correia@tecnico.ulisboa.pt}{https://orcid.org/0000-0001-7873-5531}{}

\authorrunning{N. Braz, D. Poças, M. P. Correia} 

\Copyright{Nuno Braz, Diogo Poças, and Miguel Correia} 



\category{} 

\relatedversion{} 

\acknowledgements{This work was financially supported by Project Blockchain.PT – Decentralize Portugal with Blockchain Agenda (Project no 51), WP 6: Digital Assets Management, Call no 02/C05-i01.01/2022, funded by the Portuguese Recovery and Resilience Program (PRR), The Portuguese Republic, and The European Union (EU) under the framework of the Next Generation EU Program. This work was also supported by national funds through Fundação para a Ciência e a Tecnologia, I.P. (FCT) under projects UID/50008/2025 (DOI:~\url{https://doi.org/10.54499/UID/50008/2025}) and UID/50021/2025 (DOI:~\url{https://doi.org/10.54499/UID/50021/2025}).}

\nolinenumbers 

\makeatletter
\let\@oddfoot\@empty
\makeatother

\EventEditors{John Q. Open and Joan R. Access}
\EventNoEds{2}
\EventLongTitle{42nd Conference on Very Important Topics (CVIT 2016)}
\EventShortTitle{CVIT 2016}
\EventAcronym{CVIT}
\EventYear{2016}
\EventDate{December 24--27, 2016}
\EventLocation{Little Whinging, United Kingdom}
\EventLogo{}
\SeriesVolume{42}
\ArticleNo{23}

\begin{document}

\maketitle

\begin{abstract}
Decentralized Autonomous Organizations (DAOs) increasingly delegate decisions to small committees whose members hold two independent kinds of private information: idiosyncratic preferences over alternatives (what they want) and beliefs about which alternative best serves the organization (what they know). Members have no reason to reveal what they know unless they are incentivized to do so.

Standard voting rules are designed to extract only what members want. If these two pieces of information are not aligned, the organization can end up making a suboptimal decision. Existing mechanisms for eliciting (extracting) beliefs from experts, such as Decision Scoring Rules and Decision Markets, face well-known impossibility results in multi-agent settings under deterministic decision rules, and ignore the fact that agents can have idiosyncratic preferences over alternatives.

This paper proposes a mechanism for binary committee decisions that augments the pivotal mechanism, an instance of the Groves mechanism, with a reward that depends on the outcome distributed after a boolean signal, indicating success or failure of the outcome, is observed. The mechanism aggregates the experts' private information to maximize the probability that its decision agrees with a weighted linear pooling of the experts' beliefs, framing the designer's problem as classification rather than welfare maximization.

An affine family of reward rules is proposed, and the region of parameters that simultaneously satisfy dominant-strategy incentive compatibility, interim individual rationality, and budget feasibility is derived. In informative-belief environments, the mechanism outperforms majority voting, including majority voting augmented with the same outcome-contingent rewards, especially when the committee's average preferences are biased against the superior alternative.
\end{abstract}

\section{Introduction}
\label{sec:introduction}

A Decentralized Autonomous Organization (DAO) is an organization governed by rules encoded as smart contracts on a blockchain, where members coordinate and make collective decisions without centralized authority. In recent years, DAOs have grown into large-scale entities where thousands of members coordinate pseudonymously \cite{sharma2024future}. Unlike traditional contracts, which merely articulate rules and can only serve as legal proof, smart contracts can execute the rules they encode. 
The scale of these organizations creates a significant challenge in decision-making. Several studies~\cite{cong2025centralized, rikken2019governance, kitzler2024governance} confirm that DAOs suffer from low participation rates and slow execution, showing the need for delegated governance structures. In large blockchains such as Bitcoin and Ethereum, the decision-making is done by the maintainers \cite{bip2, becze2015eip}; Polkadot elects a council \cite{polkadot_web3}; project Catalyst has a Community Advisors board and a veteran \cite{projectcatalyst_docs}. Either way, when governance councils are not elected, decision power is concentrated in a small number of members \cite{kitzler2024governance}. Empirical and analytical studies confirm that token-based voting is insufficient when participants have heterogeneous stakes and information \cite{kitzler2024governance, han2025dao, UnderstandingDAOs, dimitri2023voting}.

Restricting power to a small group or committee amplifies a natural problem: members often have personal interests that are not aligned with the organization's goal. For example, a committee governing a lending DAO might vote to approve a highly risky token as collateral simply because they personally own large amounts of it, exposing users to severe financial risk. Mechanisms such as majority voting and quadratic voting \cite{lalley2018quadratic, Flexible-Design-for-Public-Goods} are designed to account only for what members \emph{want} and not for what they \emph{know} about which alternative benefits the organization. Furthermore, any non-continuous voting mechanism necessarily loses information by restricting experts to discrete reports. Talmon \cite{talmon2023social} identifies DAOs as an important application domain for Multi-Agent Systems and calls for the community to address their governance challenges. Our work addresses this by proposing a decision-making mechanism with formal incentive guarantees that improves collective decisions through better information aggregation.

Organizations often do an evaluation of the consequences of their decisions. These evaluations are usually informal: performance reviews, retrospectives, reputation effects. In a DAO, the evaluation can be formalized: the community can express whether a decision led to a positive or negative outcome after a predetermined period. More quantitative criteria also serve, such as treasury growth, total value locked, user adoption. The evaluation task is simpler than the original decision, as it concerns observed outcomes rather than predictions. Because smart contracts can hold funds in escrow and release them conditionally, members' compensation can naturally be tied to the actual consequences of the decisions they supported.
This allows what Kiayias and Lazos~\cite{kiayias2022sok} call \emph{accountability}: a blockchain governance system satisfies the property of accountability if whenever participants bring in a change, they are held individually responsible for it in a clearly defined way by the platform.

Hanson's futarchy proposal \cite{hanson2013shall} pioneered the idea of using prediction markets to guide policy decisions, leveraging the collective intelligence of the market to evaluate which policy would produce better outcomes. Decision Scoring Rules~\cite{decision-scoring-rules, boutilier2011eliciting} and Decision Markets~\cite{decision-markets, chen2011decision} formalize it by conditioning monetary transfers on the observed outcome, but they are designed for single experts or community-wide market participation. In practice, decision markets have been deployed in DAOs either as the decision mechanism itself (MetaDAO~\cite{metadao}) or as a filtering layer that delegates the final decision to a committee (DAOstack~\cite{daostack_whitepaper}). An analysis on the latter noted that the system became self-referential: the collective intelligence of the market never showed up, and proposals were passed not because they were deemed relevant, but because their authors wished to see them approved \cite{davo2025rise}. One could envision an inverted DAOstack architecture: the committee decides, and the outcome is evaluated ex-post by a market or a peer prediction method such as the one of Miller et al. \cite{miller2005eliciting}. This paper formalizes a new mechanism structure where transfers can be conditioned on both reports and impact of the outcomes, abstracting practical considerations such as the medium of exchange, the committee selection process, and the outcome evaluation procedure, allowing the incentive analysis to apply across diverse instantiations. 


\subsection{Mechanism Overview}

We focus on binary decisions, the most common type in DAOs: approximately 67\% of analyzed DAO proposals use a binary voting pattern~\cite{UnderstandingDAOs}. 

Each committee member---hereafter \emph{expert}---possesses two independent, private types of information: \emph{idiosyncratic preferences}, how much they personally stand to gain or lose from each alternative (what they want), and \emph{subjective beliefs}, their assessment of which alternative will improve the organization's success metric (what they know). The designer wants to decide based on beliefs alone, but experts have no incentive to report their beliefs truthfully---their preferences create an incentive to misrepresent. Rather than trying to elicit (extract) preferences and beliefs separately, our mechanism induces each expert to submit a single continuous report that encodes both, and maximizes the probability that the aggregation of the reports matches the one the designer would make if all beliefs were observable.

Unlike Decision Scoring Rules, which elicit beliefs by directly scoring prediction accuracy, our mechanism ties payments to the realized outcome, so that each expert's belief is revealed through their own expected-utility maximization. Incentive compatibility comes from the pivotal structure, an instance of the Groves Mechanism~\cite{green1979incentives}, in which each agent's payment equals the externality their report imposes on the others. The role of the outcome-contingent rewards is to reshape the content of the true report, folding belief information into the scalar message that the pivotal mechanism would otherwise elicit as pure idiosyncratic preference. When experts do not have idiosyncratic preferences (a standard assumption in Decision Scoring Rules), the mechanism selects the alternative that a designer would select if given access to all private beliefs and aggregated them as a linear opinion pool, with each expert's weight equal to $\frac{b}{1 + p(a - 1)}$, where $p$ is the expert's belief that their preferred alternative will succeed and $a$ and $b$ are positive numbers that can be different between experts. These parameters are the intuitive counterparts of the reward coefficients $(\alpha, \beta)$ derived in Section~\ref{sec:mechanism}.

The mechanism sidesteps the strict-properness impossibilities of decision scoring rules under deterministic decision rules~\cite{decision-markets, decision-scoring-rules} and ensures that each expert's report is strictly optimal whenever they assign non-zero probability to any aggregate of the other agents' reports.

Several mechanisms have been proposed for aggregating preferences in binary decision settings, each with different trade-offs regarding incentive compatibility, information elicitation, and budget requirements. However, most of these mechanisms do not account for the distinction between experts' idiosyncratic preferences and their beliefs about which alternative benefits the organization.
Table~\ref{tab:decision_mechanisms} compares decision-making mechanisms along the following dimensions: whether the mechanism elicits experts' subjective beliefs about the success of each alternative (\emph{Belief Elicitation}); whether it mitigates the influence of experts' private preferences on the collective decision (\emph{Reduce Idiosyncratic Noise}); whether it supports multiple simultaneous participants (\emph{Multi-Agent}) and whether it operates in a single round (\emph{One-Shot}).
In this paper, we construct a mechanism comprising an allocation rule, a pivotal transfer rule, and an outcome-contingent reward rule that depends on a boolean signal observed after the decision is implemented. We formulate a design objective for committee decision-making---maximizing the probability that the mechanism's decision matches a \emph{linear opinion pool} of the experts' private beliefs, framing the designer's problem as classification rather than welfare maximization. We design an affine parameterization of the reward rules and derive restrictions of the feasible region of design parameters that simultaneously satisfy dominant-strategy incentive compatibility, interim individual rationality, and budget feasibility. We derive the dominant strategy, in which each expert's report jointly encodes preferences and beliefs into a single scalar, with the reward parameters controlling how strongly beliefs enter the report and how experts' messages are scaled. Finally, we evaluate the mechanism via Monte Carlo simulations, showing that when experts receive signals correlated with the true success probabilities it outperforms majority voting with outcome-contingent rewards and a DSR applied to a random expert in most scenarios.

\newcolumntype{V}{!{\vrule width 1.8pt}}
\newcommand{\fball}{\tikz[baseline=-0.5ex]\draw[fill=black] (0,0) circle (0.10cm);}
\newcommand{\eball}{\tikz[baseline=-0.5ex]\draw[thick] (0,0) circle (0.10cm);}
\newcommand{\hball}{\tikz[baseline=-0.5ex]{\draw[thick] (0,0) circle (0.10cm); \fill (90:0.10cm) arc (90:270:0.10cm) -- cycle;}}

\begin{table}[h]
\centering
\footnotesize
\renewcommand{\arraystretch}{1.1}
\begin{tabular}{ l c c c c c }
\toprule
\textbf{Feature} & \textbf{This} & \textbf{MV} & \textbf{QV} & \textbf{DSR} & \textbf{DM} \\ \midrule
{Belief Elicitation}         & \fball & \eball & \eball & \fball & \fball \\
{Reduce Idiosyncratic Noise} & \fball & \eball & \eball & \hball\rlap{\,$^{\dagger}$} & \hball\rlap{\,$^{\dagger}$} \\
{Multi-Agent}                & \fball & \fball & \fball & \eball & \fball \\
{One-Shot}                   & \fball & \fball & \fball & \fball & \eball \\
\bottomrule
\end{tabular}
\caption{Comparison of decision-making mechanisms in decentralized organizations. \fball~= satisfied, \hball~= partially satisfied, \eball~= not satisfied. \textsuperscript{$\dagger$}Outcome-contingent subsidies can indirectly mitigate idiosyncratic noise, but these mechanisms were not designed for this purpose.}
\label{tab:decision_mechanisms}
\end{table}

The remainder of the paper is organized as follows: Section~\ref{related-work} reviews related work. Section~\ref{preliminaries} presents the model, the mechanism structure, and the designer's objective. Section~\ref{sec:mechanism} progressively derives the admissible reward rules that satisfy a set of important properties. Section~\ref{sec:parameterization} analyzes the design parameters and presents simulation results. Section~\ref{sec:conclusions} concludes.

\section{Related Work}\label{related-work} 

This paper sits at the intersection of three areas: governance in decentralized autonomous organizations, information aggregation in committees, and outcome-conditioned elicitation mechanisms (decision scoring rules and decision markets). DAO governance has identified the problem we address: committees making decisions where members hold both idiosyncratic preferences and private beliefs about what best serves the organization. We review information aggregation in committees and mechanisms with similar structure to ours.

\paragraph*{Committees}

A foundational result in committee decision-making is the Condorcet Jury Theorem (1785): if each member of a group independently identifies the correct alternative with probability greater than $1/2$, then majority voting selects the correct alternative with probability approaching~$1$ as the group grows \cite{de2014essai}. The theorem assumes identical preferences and sincere voting. Austen-Smith and Banks~\cite{austensmith1996} find that honest voting by all individuals is not rational even when they have common preferences. 
Li et al.~\cite{li1999conflicts} study committee decisions where members have conflicting preferences but share a common interest in making the correct choice, an objective close to ours. In their model, preferences are publicly known while private information (a signal about which alternative is correct) must be elicited through reports. The members' preferences derive from the personal costs of incorrect decisions, and the model does not distinguish between bias manifested in prior beliefs and bias manifested in preferences. They show that manipulation incentives persist under any decision rule---even with continuous reports---so no equilibrium fully aggregates members' information, though some sharing still occurs.
Nitzan and Paroush~\cite{nitzan1982} studied the optimal weighted majority rule for a group of individuals with identical objectives but heterogeneous abilities. Each individual votes for one of two alternatives, identifying the correct one with a known probability that may differ across individuals.

Our mechanism departs from these settings by using transfers to establish a dominant strategy for each agent. The known equilibrium strategy allows the designer to provide formal incentive guarantees and to tune the mechanism's parameters to control how private information enters the aggregate, even without observing individual types.

\paragraph*{Decision Scoring Rules and Decision Markets}

A DSR pairs a scoring function that rewards a single expert based on their report and the actual outcome with a decision rule that selects an action from the report; in decision markets, experts instead produce forecasts conditional on each action. Both were first introduced by~\cite{decision-markets}.
Oesterheld and Conitzer~\cite{decision-scoring-rules} showed that DSRs can only strictly incentivize reporting the expected utility of the recommended action, not the full probability distribution over outcomes. 
Chen et al.~\cite{chen2011decision} showed that in decision markets it is impossible to strictly incentivize truthful reports about the likelihood of achieving some outcomes under deterministic decision rules; in other words, to strictly incentivize truthful reports the decision maker must risk choosing a suboptimal alternative. The authors further study the case where a single expert simply recommends an action, extending the model introduced by \cite{decision-markets}.

Our setting departs from the single-expert focus of decision scoring rules: we aim to follow the action favored by the \emph{aggregate} of the experts' beliefs rather than that of any single expert. Extending DSRs to multiple agents directly introduces strategic manipulation. As in \cite{decision-markets}, we restrict attention to two possible
outcomes, ``good'' and ``bad''. Our mechanism assigns each expert a reward that differs according to whether the realized outcome was good or bad, and charges the pivotal transfer on their report. This makes truthful reporting a dominant strategy, strictly optimal whenever an expert assigns non-zero probability to every possible aggregate of the other experts' reports. As a result, the designer can select the action it would have chosen with direct access to the experts' private beliefs under a linear pooling rule.

Considering experts with their own preferences, \cite{boutilier2011eliciting} addresses self-interested experts by introducing compensation rules that offset the expert's inherent utility from the decision. However, constructing proper compensation rules requires the principal to know the expert's utility function; when this knowledge is unavailable, properness cannot be guaranteed in general. Our mechanism does not assume knowledge of experts' preferences, though a known type distribution lets the designer anticipate its expected accuracy.

\paragraph*{Two-Stage Mechanisms}
The idea of conditioning transfers on post-decision information is not new. Mezzetti~\cite{mezzetti2004mechanism} proposed a two-stage Groves mechanism where agents report realized payoffs after the outcome, and \cite{nath2013strict} strengthened this to strict ex-post incentive compatibility. Our mechanism shares the two-stage structure but departs in two ways: the objective is classification rather than efficiency, and the second stage does not require agents to report their realized payoffs---a single boolean $\Delta \in \{-1,1\}$ is observed, eliminating the need for incentive-compatible second-stage strategies.

\section{Model}\label{preliminaries}
This section presents the model and the necessary background.
We consider a single decision (one-shot setting) in which a DAO faces a choice between two mutually exclusive alternatives $\{A, B\}$. The committee consists of a set $N = \{1, \dots, n\}$ of risk-neutral experts who aim to maximize their expected utility. 
Following the implementation of the alternative, a boolean state of the world, $\Delta \in \{-1, 1\}$ is observed, where $\Delta = 1$ denotes a positive impact and $\Delta = -1$ denotes a negative impact. Only the implemented alternative is realized, so the mechanism observes $\Delta$ for the chosen option alone and never the counterfactual outcome of the discarded one; conditioning transfers on $\Delta$ requires only the realized impact of the chosen alternative to be resolvable to a boolean against a benchmark, not that the two alternatives be compared directly. We take the procedure that produces $\Delta$---a community retrospective, a quantitative KPI threshold such as treasury growth or total value locked, or an arbitration method such as Kleros~\cite{lesaege2019kleros}---as given. Its choice is orthogonal to the incentive analysis, but the guarantees hold only if $\Delta$ is observed honestly (and experts understand how it is determined).

Experts have quasi-linear utility: a payment of~$x$ to the mechanism reduces the expert's utility by exactly~$x$, with no indirect benefit from the collective pool increasing. The medium of exchange is left unspecified---it may be tokens from the DAO treasury, reputation points, or any other transferable unit of value. We assume that transfers are enforceable, e.g. via smart contracts, and that posting collateral does not change the expert's utility function. Experts are risk neutral, thus they value the payment now as much as they will value it later. The rules of the mechanism are common knowledge. We focus on environments where the number of experts is insufficient for traditional decentralized prediction markets to function effectively.

\subsection{Information Structure}

Each expert $i \in N$ has a private type---a collection of private information that fully characterizes the expert---comprising their idiosyncratic preferences and subjective beliefs.

The expert's idiosyncratic utility is captured by the pair $v_i = (v^A_i, v^B_i) \in \mathbb{R}^2$, where $v^A_i$ and $v^B_i$ represent the expert's personal utility derived from the implementation of option A and option B, respectively. Since exactly one of these options will be implemented, the expert can normalize their preference parameters accordingly so without loss of generality, at least one of these alternatives yields a non-negative idiosyncratic utility for the expert. 

The expert's subjective beliefs are represented by the pair $p_i = (p^A_i, p^B_i) \in \mathcal{P}^2$, where $\mathcal{P} = (0, 1)$ and $p^A_i$ and $p^B_i$ denote the expert's private assessment of the probability that alternatives A and B will lead to a positive outcome for the organization, respectively. We exclude the cases where the experts are certain, to avoid unnecessary notation.
They are formally defined as the conditional probabilities of observing a positive outcome given the chosen alternative: $p^A_i = \mathbb{P}_i(\Delta > 0 \mid A)$ and $p^B_i = \mathbb{P}_i(\Delta > 0 \mid B)$.
These are subjective conditionals: they may differ across experts and generally will not coincide with the true success probabilities.

The expert’s type is given by $\theta_i \in \Theta_i$, where $\theta_i = (v^A_i, v^B_i, p^A_i, p^B_i)$. We denote the type profile of all experts’ types by $\mathbf{\theta} = (\theta_1, \dots, \theta_n)$ and the belief profile by $\mathbf{p} = (p_1, \dots, p_n)$ with $p_i = (p^A_i, p^B_i)$. We operate under the Independent Private Values (IPV) model: each expert’s type $\theta_i$ is drawn independently from some distribution over the type space $\Theta_i$.

To clarify this assumption, we contrast our setting with the mineral rights auction \cite{kagel1986winner}. There, each bidder receives a private geological signal, but the true value of the tract is a common function of all signals. Learning another bidder’s signal would cause a bidder to revise their own valuation. A correlated values model is necessary precisely because no public deliberation stage resolves the common-value uncertainty before bids are placed. In DAO governance the mechanism is preceded by forum discussions, temperature checks, and community calls that generate a common public signal observed by all experts. Each expert absorbs this shared information and updates accordingly. We capture what stays private after discussion. In this way, an expert’s private signal carries no information about another expert’s signal: learning it would not cause a further update. 

We do not assume that the two pairs $(v^A_i, v^B_i)$ and $(p^A_i, p^B_i)$ are correlated. Experts tend to benefit from the options they believe will succeed which aligns private incentives with informational contributions and makes the designer’s task easier. Independence is therefore a conservative assumption.

The mechanism takes the committee as given: the selection procedure determines robustness to coordination among experts. The properties rely on independent private values; when experts coordinate before reporting---through shared information, bribery, or vote-buying---this assumption breaks down.

\subsection{Mechanism Structure}
A direct revelation mechanism asks each participant to report their private information and then determines an outcome and monetary transfers as a function of the submitted messages. Formally, it consists of two components: an \emph{allocation rule}, which maps reports to an outcome, and a \emph{transfer rule}, which determines the payment each participant makes or receives based on the reports. We introduce a third component, a \emph{reward rule}, that distributes payments after the outcome is observed. The reward rule is defined separately from the transfer rule because the expert evaluates it only in expectation over the unknown outcome at the time of reporting, while the transfer is known. This separation allows us to restrict the admissible reward rules independently of the allocation and transfer rules.

In our case, the mechanism proceeds in two stages, illustrated in Figure~\ref{fig:mechnaism_overview}. First, each expert submits a message encoding how strongly the expert favors one alternative over the other. The mechanism aggregates these messages, selects an alternative, and charges each expert a transfer that depends only on the submitted messages. Second, after the outcome is observed, the mechanism distributes a reward that depends on whether the decision led to a positive or negative outcome. The expert's payoff is therefore determined by three factors: which option is chosen, how much they pay upfront, and what reward they receive once the outcome is revealed.
\begin{figure}[htbp]
    \centering
    \includegraphics[width=0.8\textwidth, trim=5cm 12cm 0cm 9cm, clip]{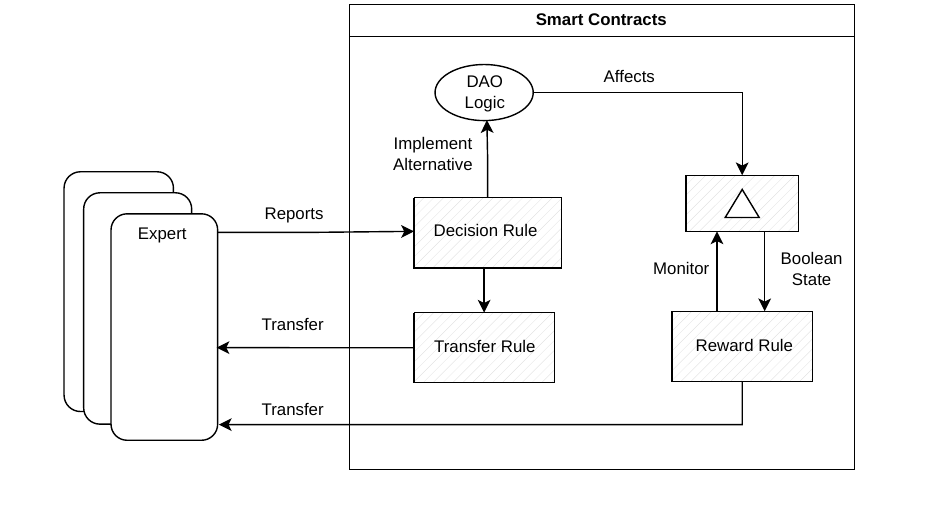}
    \caption{Mechanism Overview}
    \label{fig:mechnaism_overview}
\end{figure}

All experts share a common message space $\mathcal{R}$. The profile of messages submitted by all experts is denoted by $\mathbf{m} = (m_1, \dots, m_n) \in \mathcal{R}^n$. We use the notation $\mathbf{m}_{-i}$ for the profile of all experts except~$i$, and sometimes write $\mathbf{m} = (m_i, \mathbf{m}_{-i})$.

We formally define a mechanism in this setting as a triplet $\mathcal{M} = (x, t, r)$. We use $\mathbf{t},\mathbf{r}$ to refer to vectors denoting the application of each rule:
\begin{itemize}
    \item $x:  \mathcal{R}^n \to \{A, B\}$ is the deterministic allocation rule representing the collective decision.
    \item $\mathbf{t} = (t_1, \dots, t_n)$ is the transfer rule, where $t_i: \mathcal{R}^n\to \mathbb{R}$ represents a monetary transfer to expert $i$ based solely on the submitted messages. A positive transfer means the expert receives money, while a negative transfer means they pay.
    \item $\mathbf{r} = (r_1, \dots, r_n)$ is the reward rule, where $r_i: \mathcal{R}^n \times \{-1, 1\} \to \mathbb{R}$ represents the ex-post reward distributed to expert $i$, contingent on the realized state of the world $\Delta$.
\end{itemize}

Given a type $\theta_i = (v^A_i, v^B_i, p^A_i, p^B_i)$ and a profile of submitted messages $\mathbf{m}$, let $o = x(\mathbf{m}) \in \{A, B\}$ denote the chosen alternative. We write $v^o_i$, $p^o_i$ for the expert's utility and belief associated with the chosen alternative. The expected utility of expert $i$ reflects the three factors described above:
\begin{equation}\label{expert-utility}
    \mathbb{E}[u_i(\mathbf{m} \mid \theta_i)] =
    v^{o}_i  + t_i(\mathbf{m})+ \mathbb{E}_{\Delta}[r_i(\mathbf{m}, \Delta) \mid p^{o}_i]
\end{equation}

\subsection{Background: The Pivotal Mechanism}\label{pivotal}
Our mechanism builds on the Pivotal Mechanism \cite{green1979incentives}. Each expert submits a scalar $m_i \in \mathbb{R}$ representing the intensity of their preference. The allocation rule selects $A$ if the sum of messages is non-negative, otherwise it selects $B$:
\begin{equation}\label{allocation-rule}
x(\mathbf{m}) = \begin{cases} A & \text{if } \sum_{i=1}^n m_i \geq 0 \\ B & \text{otherwise} \end{cases}
\end{equation}
The transfer charges each pivotal expert the externality they impose on others:
\begin{equation} \label{Pivotal}
t_i(\mathbf{m}) = \begin{cases}
-\sum_{j\neq i} m_j & \text{if } -m_i >\sum_{j\neq i} m_j >0\\[6pt]
\sum_{j\neq i} m_j  & \text{if } -m_i<\sum_{j\neq i} m_j  <0 \\[6pt]
0 &\text{otherwise}
\end{cases}
\end{equation}
Non-pivotal experts pay nothing and reporting $m_i = v_i^A - v_i^B$ is a dominant strategy.

In the pivotal mechanism, each expert has a \emph{winning set}: the set of profiles where their preferred alternative is selected. An expert who prefers~$A$ ($m_i > 0$) wins when $\sum_j m_j \geq 0$; submitting a larger positive message pushes the aggregate toward their winning set. An expert who prefers~$B$ ($m_i < 0$) wins when $\sum_j m_j < 0$; submitting a more negative message pushes toward theirs. The \emph{critical value} is the minimum message needed to enter the winning set: $c_i = -\sum_{j \neq i} m_j$. The pivotal transfer~\eqref{Pivotal} charges exactly this value to pivotal experts. By Myerson's lemma~\cite{myerson1981optimal}, a mechanism is DSIC if and only if the allocation is monotone in the message and each agent pays their critical value. DSIC is characterized in terms of the sum of both transfers: $t_i(\mathbf{m}) + \mathbb{E}_\Delta[r_i(\mathbf{m},\Delta)]$.

\subsection{Designer's Objective}\label{designer-objective}
Unlike most mechanisms that seek to maximize the sum of idiosyncratic preferences $v_i$, here these valuations are exactly the idiosyncratic noise the designer wants to filter out. The designer instead wants to select the alternative most likely to yield a positive outcome for the organization, which is given by the experts' beliefs, a quantity no agent wants to report directly. 

The designer's goal is to select the alternative that an informed designer would choose if given direct access to all experts' private beliefs about each alternative's success probability. Let $x^* : \mathcal{P}^{2n} \to \{A,B\}$ denote an ideal decision rule --- the allocation the designer would choose if all beliefs were observable. However, the designer does not observe $p_i$; instead, the mechanism collects messages $m_i$ that encode both $v_i$ and $p_i$. The designer seeks to maximize the probability that the mechanism's decision agrees with the ideal one:
\begin{equation}\label{eq:objective}
\max_{x, t, r} \; \mathbb{P}\!\left(x(\mathbf{m}^*) = x^*(\mathbf{p})\right),
\end{equation}
where $\mathbf{m}^* = (m^*(v_1, p_1), \ldots, m^*(v_n, p_n))$ is the message profile under the dominant strategy $m^*$ induced by $(x, t, r)$, and the probability is over the random types $(v_i, p_i)$.

A \emph{strategy} for agent $i$ is a function $m : \mathbb{R}^2 \times \mathcal{P}^2 \to \mathbb{R}$ mapping the agent's type $(v_i, p_i)$ to a message $m_i = m(v_i, p_i)$. The designer seeks a mechanism in which each agent has a dominant strategy that depends on both components of the type. Since the mechanism's instruments $(x, t, r)$ cannot depend on $p_i$ directly, the dominant strategy must be \emph{implementable}: it must arise from the agent's optimization under rules independent of $p_i$. 

We instantiate the general objective by choosing an ideal decision rule $x^*$ based on \emph{linear pooling}~\cite{winkler2019probability, ranjan2010combining}. Linear pooling is a well-studied method for combining expert forecasts: it averages the individual probability assessments with weights set by the designer. In our setting, this amounts to selecting $A$ whenever the experts' average belief differential favors it:
\begin{equation}\label{ideal-rule}
x^*(\mathbf{p}) = \begin{cases} A & \text{if } \sum_{i=1}^n w^A_ip^A_i - w^B_ip^B_i > 0 \\ B & \text{otherwise} \end{cases}
\end{equation}
An alternative is logarithmic pooling, which gives disproportionate weight to extreme assessments: if one expert assigns 90\% to~$A$ while five others assign 40\%, logarithmic pooling selects~$A$, whereas linear pooling selects~$B$.

Since the designer does not observe $(p^A_i, p^B_i)$, the mechanism must rely on the messages $m_i$. The designer therefore seeks a reward rule $r$ such that the sign of the aggregate message matches the sign of the aggregate belief differential:
\begin{equation}\label{sign-match}
\text{sgn}\left(\sum_{i=1}^{n} m_i\right) = \text{sgn}\left(\sum_{i=1}^{n} w^A_ip^A_i - w^B_ip^B_i\right), \quad \forall \theta \in \Theta
\end{equation}
subject to the constraints defined in the next section.
We consider equal weights throughout the rest of the paper.

\subsection{Formalization of the Properties}\label{Properties}

The desired properties for the decision-making mechanism are formalized below. 

The mechanism must remain invariant to the relabeling of alternatives. Since the reports are reflected by the experts' utility functions we formalize this requirement as:

\newtheorem{property}{Property}

\begin{property}[Symmetry]\label{def:symmetry}
A mechanism satisfies Symmetry if, for every agent $i$, every type $\theta_i = (v^A_i, v^B_i, p^A_i, p^B_i)$ with relabeled type $\tilde{\theta}_i = (v^B_i, v^A_i, p^B_i, p^A_i)$:
\[\mathbb{E}[u_i(\mathbf{m} \mid \theta_i)] = \mathbb{E}[u_i(-\mathbf{m} \mid \tilde{\theta}_i)] \]
where $-\mathbf{m}$ is the profile obtained by negating each message.
\end{property}

Next, we establish the condition under which experts are willing to participate in the DAO's decision-making process, formally known as Individual Rationality (IR). We adopt the notion of interim IR: each expert knows their own preferences and beliefs when deciding whether to vote, but does not know what the other experts will report. Additionally, there also exist the notions of ex-ante (before learning one's type) and ex-post (after all messages are submitted). 

\begin{property}[Interim Individual Rationality]\label{def:IR}
    A decision-making mechanism satisfies satisfies Interim Individual Rationality (IR) with reservation level $\bar{u}$ if:
        \[ \forall i \in N, \forall \theta_i \in \Theta_i, \forall \mathbf{m}_{-i} \in \mathcal{R}_{-i}: \quad \mathbb{E}[u_i(m_i^*, \mathbf{m}_{-i} \mid \theta_i)] \geq \mathbb{E}[u_i(\mathbf{m}_{-i} \mid \theta_i)] + \bar{u} \]
        for some $m_i^* \in \mathcal{R}$.
\end{property}

A committee decision is \emph{non-excludable}: expert~$i$ is bound by the collective outcome whether or not they vote. The relevant outside option is therefore not zero utility but \emph{abstention} --- the expert declines to influence the vote while the others still decide, receiving $\mathbb{E}[u_i(\mathbf{m}_{-i}\mid\theta_i)]$, in which $t_i = 0$, the reward vanishes, and they obtain only their valuation $v_i^{x(\mathbf{m}_{-i})}$ of the outcome chosen by others. Participation is individually rational when the expected gain over this abstention payoff clears a reservation level $\bar{u} \leq 0$, which measures how far below the abstention payoff an expert will still tolerate in order to keep their seat. When exit is costless, an expert participates only if doing so is at least as good as abstaining, so $\bar{u} = 0$. 

Having established the conditions for expert participation, we must also consider the financial viability of the mechanism from the perspective of the DAO. A mechanism satisfies the Budget Constraint if the total net compensation issued by the DAO to the experts is bounded by a maximum subsidy. This bound applies regardless of the medium of exchange: even when transfers are denominated in reputation points rather than tokens, the mechanism is constrained by the maximum value that experts attribute to reputation. 
\begin{property}[Budget Feasibility]
    A decision-making mechanism satisfies Budget Feasibility with budget $c \geq 0$ if:
    \[ \forall \mathbf{m} \in \mathcal{R}^n, \forall \Delta \in \{-1, 1\}: \quad \sum_{i \in N} \Big( t_i(\mathbf{m}) + r_i(\mathbf{m}, \Delta) \Big) \leq c \]
\end{property}


A mechanism is incentive compatible if experts maximize their expected utility by honestly revealing this information rather than behaving strategically. A DSIC mechanism guarantees that reporting honestly is an expert's best strategy regardless of the reports of other participants.
\begin{property}[Dominant Strategy Incentive Compatibility]
A decision-making mechanism is Dominant Strategy Incentive Compatible (DSIC) if,
    \[ \forall i \in N, \forall \theta_i \in \Theta_i, \forall m_i \in \mathcal{R}, \forall \mathbf{m}_{-i} \in \mathcal{R}^{n-1}: \]
    \[ \mathbb{E}[u_i(m_i^*, \mathbf{m}_{-i} \mid \theta_i)] \geq \mathbb{E}[u_i(m_i, \mathbf{m}_{-i} \mid \theta_i)] \]
    where $m_i^*$ denotes the honest reporting strategy for type $\theta_i$.
\end{property}

Ideally, the designer would observe each expert's full type~$\theta_i = (v_i, p_i)$, discard the preferences~$v_i$, and decide based on the beliefs~$p_i$ alone. But experts have no incentive to report~$p_i$ truthfully---their preferences create an incentive to misrepresent their beliefs. Rather than eliciting~$v_i$ and~$p_i$ separately, we design a mechanism that induces experts to report a single number~$m_i^*$ that encodes both, in a way that brings the collective decision closer to the designer's goal. What counts as ``truthful'' depends on the question being asked. Throughout the paper, we use DSIC to mean that the mechanism has a dominant strategy~$m_i^*(\theta_i)$. By the revelation principle, this is without loss of generality~\cite{gibbard1973manipulation}.

\section{Mechanism Design Space}\label{sec:mechanism}

The mechanism uses the pivotal transfer from Section~\ref{pivotal} and augments it with an outcome-contingent reward rule. The protocol proceeds as follows: (1)~each expert deposits a stake into the governance smart contract; (2)~each expert submits a message $m_i \in \mathbb{R}$; (3)~the allocation rule selects an alternative; (4)~the pivotal transfer is deducted from each expert's escrow; (5)~after a predetermined evaluation period, the on-chain tool observes $\Delta \in \{-1,1\}$; (6)~the reward is distributed and remaining escrow is returned.

We restrict attention to \emph{anonymous} reward rules, where every expert faces the same rule $r_i = r$. Non-anonymous rules, assigning per-expert parameters, let the designer attribute different weights to experts---prioritizing those with demonstrated expertise---without affecting the incentive guarantees. The properties from Section~\ref{Properties} progressively narrow the space of admissible reward rules.


\paragraph*{Symmetry}\label{subsec:Symmetry}

The reward rule has four components: $r_A^+$, $r_A^-$, $r_B^+$, and $r_B^-$, corresponding to the two alternatives and two outcome realizations. These are expert~$i$'s reward functions; throughout this paragraph we suppress the expert subscript. Since the mechanism must not favor either alternative, the reward for choosing $A$ under a given profile must equal the reward for choosing $B$ under the relabeled profile. This reduces the design space from four functions to two:

\begin{proposition}\label{thm:Symmetry}
The mechanism $(x, t, r)$ satisfies Symmetry if and only if
\[r_A^+(\mathbf{m}) = r_B^+(-\mathbf{m}), \qquad r_A^-(\mathbf{m}) = r_B^-(-\mathbf{m}) \qquad \forall\, \mathbf{m}\]
\end{proposition}

\begin{proof}
Let $o=x(\mathbf{m})$ denote the chosen alternative under a message profile $\mathbf{m}$ and $\bar{o} = x(-\mathbf{m})$. 
We have
\[
\mathbb{E}[u_i(\mathbf{m} \mid \theta_i)] = v_i^{o} + t_i(\mathbf{m}) + p_i^{o}\, r_o^+(\mathbf{m}) + (1 - p_i^{o})\, r_o^-(\mathbf{m})
\]
\[
\mathbb{E}[u_i(-\mathbf{m} \mid \tilde{\theta}_i)] = v_i^{\bar{o}} + t_i(-\mathbf{m}) + p_i^{\bar{o}}\, r_{\bar{o}}^+(-\mathbf{m}) + (1 - p_i^{\bar{o}})\, r_{\bar{o}}^-(-\mathbf{m})
\]
By the neutrality of the allocation rule (negating all messages selects the opposite alternative: $\bar{o} = x(-\mathbf{m})$ is the opposite of $o = x(\mathbf{m})$) and definition of the pivotal transfer, we have, $v_i^{o} + t_i(\mathbf{m}) = v_i^{\bar{o}} + t_i(-\mathbf{m})$.
Since the relabeled type $\tilde{\theta}_i$ swaps $p^A_i \leftrightarrow p^B_i$, we have $p_i^{\bar{o}}$ under $\tilde{\theta}_i$ equals $p_i^{o}$ under $\theta_i$. Thus $\mathbb{E}[u_i(\mathbf{m} \mid \theta_i)] = \mathbb{E}[u_i(-\mathbf{m} \mid \tilde{\theta}_i)]$ iff  $p_i^{o}\, r_o^+(\mathbf{m}) + (1 - p_i^{o})\, r_o^-(\mathbf{m}) = p_i^{o}\, r_{\bar{o}}^+(-\mathbf{m}) + (1 - p_i^{o})\, r_{\bar{o}}^-(-\mathbf{m})$ for all $p_i^o \in (0,1)$. This holds iff $r_o^+(\mathbf{m}) = r_{\bar{o}}^+(-\mathbf{m})$ and $r_o^-(\mathbf{m}) = r_{\bar{o}}^-(-\mathbf{m})$, which gives $r_A^+(\mathbf{m}) = r_B^+(-\mathbf{m})$ and $r_A^-(\mathbf{m}) = r_B^-(-\mathbf{m})$ for all $\mathbf{m}$.

\end{proof}
Symmetry thus collapses the four reward functions to two. We define $r_i^+(\mathbf{m}) := r_A^+(\mathbf{m})$ and $r_i^-(\mathbf{m}) := r_A^-(\mathbf{m})$, so that
\[
r_i(\mathbf{m}, \Delta) = \begin{cases} r_i^+(\mathbf{m}) & \text{if } \Delta = 1 \\ r_i^-(\mathbf{m}) & \text{if } \Delta = -1 \end{cases}
\]
It remains to determine how $r_i^+$ and $r_i^-$ may depend on the message profile. We write them as functions of~$m_i$ and~$\mathbf{m}_{-i}$; whenever the subscript~$i$ appears explicitly among the arguments, we drop it in the function name. The \emph{expected transfer} of expert~$i$---the transfer plus the reward evaluated in expectation over~$\Delta$ according to the expert's belief~$p_i^o$---is:
\begin{equation}\label{eq:expected-transfer-def}
\hat{t}_i(\mathbf{m}) \;=\; t_i(\mathbf{m}) + p_i^{o}\, r^+(m_i, \mathbf{m}_{-i}) + (1 - p_i^{o})\, r^-(m_i, \mathbf{m}_{-i}).
\end{equation}

\paragraph*{Budget Feasibility}\label{subsec:BudgetConstraint}

The Budget Feasibility restricts how $r^+$ and $r^-$ can grow with~$|t_i|$:

\begin{proposition}\label{prop:affine-form}
If the mechanism $(x, t, r)$ satisfies the Budget Feasibility and the message space is unrestricted ($m_i \in \mathbb{R}$), then $r^+$ and $r^-$ are bounded above by an affine function of~$|t_i|$.
\end{proposition}

\begin{proof}
For any fixed $\mathbf{m}_{-i}$, the transfer $t_i$ can be made arbitrarily large in magnitude by varying $m_i$. If $r^+$ (or $r^-$) grows faster than $|t_i|$ in the positive direction, the total payout $\sum_{i \in N}( t_i + r^{\pm}_i)$ grows without bound, violating the Budget Feasibility. Hence the positive part of $r^+$ and $r^-$ is at most affine in~$|t_i|$.
\end{proof}

The proposition bounds the reward from above but not from below, and does so for each fixed $\mathbf{m}_{-i}$ separately. A priori, the admissible form is therefore:
\begin{align*}
r^+(m_i, \mathbf{m}_{-i}) &= \big(\alpha^{+}(\mathbf{m}_{-i}) - 1\big)\, t_i(m_i, \mathbf{m}_{-i}) + \beta^{+}(\mathbf{m}_{-i}) - g^+(m_i, \mathbf{m}_{-i}), \\
r^-(m_i, \mathbf{m}_{-i}) &= \big(\alpha^{-}(\mathbf{m}_{-i}) - 1\big)\, t_i(m_i, \mathbf{m}_{-i}) + \beta^{-}(\mathbf{m}_{-i}) - g^-(m_i, \mathbf{m}_{-i})
\end{align*}
where $g^+, g^- \geq 0$ are arbitrary non-negative functions representing additional penalties beyond the affine component. Dominant-strategy incentive compatibility (Theorem~\ref{thm:dsic}) requires the induced transformation to be independent of~$\mathbf{m}_{-i}$, which forces $\alpha^\pm$ and $\beta^\pm$ to be constant in~$\mathbf{m}_{-i}$; otherwise the dominant strategy would depend on the other experts' reports. We set $g^+ = g^- = 0$, restricting to affine rewards, which implement the linear opinion pool that matches our instantiation of the designer's goal. Non-affine rewards remain DSIC whenever they correspond to a strictly increasing message transformation (Lemma~\ref{lem:transformation}). These functions might be used to implement a nonlinear aggregation rule; characterizing the non-affine family is left to future work. The resulting affine parameterization is:
\begin{equation}\label{eq:reward_param}
r^+(m_i, \mathbf{m}_{-i}) = (\alpha^{+} - 1)\, t_i(m_i, \mathbf{m}_{-i}) + \beta^{+} \qquad
r^-(m_i, \mathbf{m}_{-i}) = (\alpha^{-} - 1)\, t_i(m_i, \mathbf{m}_{-i}) + \beta^{-}
\end{equation}
where $\alpha^+, \alpha^-, \beta^+, \beta^- \in \mathbb{R}$ are design parameters. The coefficient $(\alpha^\pm - 1)$ measures how much the reward amplifies or attenuates the pivotal transfer when the outcome is positive or negative, and $\beta^\pm$ is a fixed payment distributed independently of the value transfer.

Substituting~\eqref{eq:reward_param} into~\eqref{eq:expected-transfer-def} and writing $\alpha(p) = \alpha^+ p + \alpha^-(1-p)$ and $\beta(p) = \beta^+ p + \beta^-(1-p)$, the expected transfer simplifies to:
\begin{equation}\label{eq:transfer-decomposition}
\hat{t}_i(\mathbf{m}) = \alpha(p_i^{o}) \cdot t_i(\mathbf{m}) + \beta(p_i^{o}),
\end{equation}
where $o = x(\mathbf{m})$ is the chosen alternative.

The $\beta$ term acts as a belief-dependent shift of the expert's valuation; write $\hat{v}_i := v_i + \beta(p^A_i) - \beta(p^B_i)$ for the resulting \emph{effective valuation}, where $v_i = v^A_i - v^B_i$. The expert behaves as a pivotal agent with valuation $\hat{v}_i$ in place of $v_i$.

The four parameters $(\alpha^+, \alpha^-, \beta^+, \beta^-)$ of the mechanism class~\eqref{eq:reward_param} constitute the design space. The budget feasibility further restricts this space:

\begin{proposition}\label{prop:alpha-cap}
Under the Budget Feasibility, the design parameters satisfy $\alpha^+ \geq 0$, $\alpha^- \geq 0$, $\beta^+ \leq c/n$, and $\beta^- \leq c/n$.
\end{proposition}

\begin{proof}
Substituting~\eqref{eq:reward_param} into the Budget Feasibility for $\Delta = 1$:
\[
\sum_{i \in N} \Big( t_i(\mathbf{m}) + (\alpha^+ - 1)\, t_i(\mathbf{m}) + \beta^+ \Big) = \alpha^+\! \sum_{i \in N} t_i(\mathbf{m}) + n\beta^+ \;\leq\; c \qquad \forall\, \mathbf{m}.
\]
If the expert is pivotal, $t_i(\mathbf{m}) \leq 0$ for all~$i$, and $\sum_{i} t_i(\mathbf{m})$ can be made arbitrarily negative by increasing the magnitude of the messages. If $\alpha^+ < 0$, the term $\alpha^+ \sum_i t_i$ is positive and unbounded, violating the constraint. Hence $\alpha^+ \geq 0$. Given $\alpha^+ \geq 0$, the left-hand side is maximized when $\sum_i t_i = 0$ (no expert is pivotal), yielding $n\beta^+ \leq c$, i.e., $\beta^+ \leq c/n$. The same argument applied to $\Delta = -1$ gives $\alpha^- \geq 0$ and $\beta^- \leq c/n$.
\end{proof}

It remains to determine which parameter values $(\alpha^+, \alpha^-, \beta^+, \beta^-)$ of the mechanisms defined by ~\eqref{eq:reward_param} yield DSIC and IR, and what dominant strategy they induce.

\paragraph*{Dominant Strategy Incentive Compatibility}\label{subsec:dsic}

The following lemma shows that applying a strictly increasing transformation to the message space of a DSIC mechanism preserves incentive compatibility, with the critical value transformed by the inverse.

\begin{lemma}\label{lem:transformation}
Let $(x, t)$ be a DSIC mechanism with winning set $W_i$ and critical value $c_i = \inf\{m_i : x(m_i, \mathbf{m}_{-i}) \in W_i\}$. If the designer wants to coerce an agent to submit $\varphi(v_i)$ instead of $v_i$, where $\varphi : \mathcal{R} \to \mathcal{R}$ is strictly increasing, then the transformed mechanism is DSIC with critical value $\varphi^{-1}(c_i)$.
\end{lemma}

\begin{proof}
The original mechanism is DSIC, so the allocation rule $x(v_i, \mathbf{m}_{-i})$ is monotonically increasing in $v_i$ and agents pay their critical value when in the winning set. Under the transformation, the allocation becomes $x(\varphi(v_i), \mathbf{m}_{-i})$, which is increasing in $v_i$ since $\varphi$ is strictly increasing. The new critical value is
\[
c_i' = \inf\{v_i : \varphi(v_i) \in W_i(\mathbf{m}_{-i})\} = \inf\{v_i : \varphi(v_i) \geq c_i\} = \varphi^{-1}(c_i).
\]
Since monotonicity is preserved and with critical value $\varphi^{-1}(c_i)$, the transformed mechanism is DSIC.
\end{proof}
In other words, if the designer applies a strictly increasing transformation $\varphi$ to the critical value, then the mechanism remains DSIC and the agents have dominant strategy $\varphi^{-1}(v)$.

\begin{theorem}\label{thm:dsic}
Let $\alpha^+, \alpha^- \geq 0$ with $(\alpha^+, \alpha^-) \neq (0, 0)$. The mechanism $(x, t, r)$ is DSIC with dominant strategy
\begin{equation}\label{eq:ds-dao}
m^*(v_i, p_i) = \begin{cases} \displaystyle\frac{v_i + \beta(p^A_i) - \beta(p^B_i)}{\alpha(p^A_i)} & \text{if } v_i + \beta(p^A_i) - \beta(p^B_i) > 0 \\[10pt] \displaystyle\frac{v_i + \beta(p^A_i) - \beta(p^B_i)}{\alpha(p^B_i)} & \text{if } v_i + \beta(p^A_i) - \beta(p^B_i) < 0 \end{cases}
\end{equation}
where $v_i = v^A_i - v^B_i$.
\end{theorem}

\begin{proof}
Fix expert~$i$ with beliefs $p_i = (p^A_i, p^B_i)$. The $\beta$ term shifts the expert's effective valuation to $\hat{v}_i$; the $\alpha$ term scales the transfer by $\alpha(p_i^o)$. When $\hat{v}_i > 0$ the expert prefers~$A$, so being in the winning set means $o = A$ and the scaling is $\alpha(p^A_i)$; when $\hat{v}_i < 0$ the expert prefers~$B$, giving $\alpha(p^B_i)$. Define the transformation
\[
\varphi(\hat{v}_i) = \begin{cases} \hat{v}_i / \alpha(p^A_i) & \text{if } \hat{v}_i > 0 \\ \hat{v}_i / \alpha(p^B_i) & \text{if } \hat{v}_i < 0 \end{cases}
\]
which is strictly increasing when $\alpha(p^A_i), \alpha(p^B_i) > 0$. By Lemma~\ref{lem:transformation}, the mechanism is DSIC with dominant strategy $m_i^* = \varphi(\hat{v}_i)$, yielding~\eqref{eq:ds-dao}.
\end{proof}
Note that $\varphi$ must depend only on the expert's own type, not on~$\mathbf{m}_{-i}$: if it did, the dominant strategy would depend on~$\mathbf{m}_{-i}$ as well, and the mechanism would no longer be DSIC. The dominant strategy is weakly optimal for all type profiles and strictly optimal whenever the expert assigns non-zero probability to being pivotal, i.e., to the event that their report changes the collective decision.

When experts have no idiosyncratic preferences, the mechanism reduces to pure belief elicitation, let, $\Delta\beta := \beta^+ - \beta^-$, 

\begin{corollary}
When $v_i = 0$ for all $i$ and $\Delta\beta > 0$, the dominant strategy reduces to $m^*_i = \Delta\beta\,(p^A_i - p^B_i)/\alpha(p^o_i)$. In particular, taking $\alpha(p^o_i) = 1$ (e.g.\ $\alpha^+ = \alpha^- = 1$) recovers the equal-weight linear opinion pool of~\eqref{ideal-rule} exactly; other admissible $\alpha$ yield a belief-weighted pool with weights $1/\alpha(p^o_i)$.
\end{corollary}
If the designer assigns non-anonymous reward rules, each expert~$i$ receives a per-expert $\Delta\beta_i$, which lets the designer weight experts by criteria such as expertise. Taking $\alpha = 1$, $m^*_i = \Delta\beta_i\,(p^A_i - p^B_i)$, the ratio $\Delta\beta_i / \Delta\beta_j$ fixes how much more expert~$i$'s belief differential counts than expert~$j$'s in the aggregate. For instance, setting $\Delta\beta_i = 2\,\Delta\beta_j$ makes expert~$i$ pivotal-equivalent to two experts~$j$ holding the same belief.

\paragraph*{Interim Individual Rationality}\label{subsec:ir}

Using the reservation level $\bar{u}$ from Property \ref{def:IR}.

\begin{proposition}\label{prop:ir}
The mechanism satisfies Interim IR if and only if $\beta(p) \geq \bar{u}$ for all $p \in (0,1)$. 
\end{proposition}

\begin{proof}
Under DSIC, truthful reporting weakly dominates any alternative report, including any report that makes expert non-pivotal. It therefore suffices to verify that the non-pivotal deviation satisfies IR; the truthful (possibly pivotal) report can only improve on it. When non-pivotal, $t_i = 0$ and the allocation is unchanged. The participation gain over $v_i^{x(\mathbf{m}_{-i})}$ is $\beta(p^{x(\mathbf{m})}_i)$. IR requires this to be at least $\bar{u}$ for all $p \in (0,1)$. Since $\beta(p) = \beta^+ p + \beta^-(1-p)$ is linear in $p$, the condition $\beta(p) \geq \bar{u}$ for all $p \in (0,1)$ holds if and only if $\beta^+, \beta^- \geq \bar{u}$. 
\end{proof}

The four design parameters $(\alpha^+, \alpha^-, \beta^+, \beta^-)$ are constrained by the mechanism's structural and incentive properties. The resulting feasible region restrictions, are summarized in Table~\ref{tab:feasible}. The mechanisms of the form~\eqref{eq:reward_param} leave the designer free to choose the following parameters $(\alpha^+, \alpha^-, \beta^+, \beta^-)$, where the ratio $\alpha^+/\alpha^-$ controls the reweighting of experts valuation and the $\beta$ parameters determine the strength of the belief-based incentive. 

\begin{table}[htbp]
\centering
\begin{tabular}{lll}
\toprule
\textbf{Property} & \textbf{Constraint} & \textbf{Reference} \\
\midrule
Budget Feasibility & $\alpha^+, \alpha^- \geq 0$ & Proposition~\ref{prop:alpha-cap} \\
                  & $\beta^+, \beta^- \leq c/n$ & \\
DSIC              & $(\alpha^+, \alpha^-) \neq (0,0)$ & Theorem~\ref{thm:dsic} \\
Interim IR        & $\beta^+, \beta^- \geq \bar{u}$ & Proposition~\ref{prop:ir} \\
\midrule
\textbf{Feasible region restrictions} & $\alpha^+, \alpha^- \geq 0,\; (\alpha^+, \alpha^-) \neq (0,0)$ & \\
                          & $\beta^+, \beta^- \in [\bar{u},\; c/n]$ & \\
\bottomrule
\end{tabular}
\caption{Parameter restrictions from each mechanism property.}
\label{tab:feasible}
\end{table}

\paragraph*{A note on Complexity}\label{subsec:complexity}

The mechanism is computationally lightweight. The allocation rule computes a single sum $\sum_{i=1}^n m_i$ in $O(n)$ arithmetic operations. The Pivotal transfer for each expert requires evaluating $\sum_{j \neq i} m_j = \sum_{i=1}^n m_i - m_i$, which is $O(1)$ per expert given $\sum_{i=1}^n m_i$, yielding $O(n)$ total. The reward rule performs a constant number of operations per expert, thus the entire mechanism therefore runs in $O(n)$ time. The remaining cost is the evaluation of $\Delta$, a single evaluation per decision.
From a smart contract perspective, the on-chain cost is equally modest. The contract stores $n$ messages and the budget $c$, computes the sum, determines the allocation, and records the Pivotal transfers---all in a single transaction with $O(n)$ storage writes and arithmetic. The ex-post reward is computed in a second transaction after $\Delta$ is resolved, requiring $O(n)$ operations. The only overhead beyond a standard voting mechanism is the evaluation of~$\Delta$, which depends on the choice of the evaluation tool, not on the mechanism itself.

\section{Parameter Analysis}\label{sec:parameterization}

This section studies how the choice of parameters affects the mechanism's decisions, how the $\beta$ terms control the strength of the belief signal and how the ratio $\alpha^+/\alpha^-$ reweights experts according to their beliefs. We further reformulate the designer's choice as an optimization problem over the mechanisms of the form~\eqref{eq:reward_param}. The resulting objective has no closed form, so we evaluate it through Monte Carlo simulation, comparing the mechanism against majority voting with and without outcome-contingent rewards and a DSR applied to a random expert across committee sizes, preference biases, and belief-noise regimes.

\subsection{Parameters}\label{subsec:parameters}

Each parameter pair shapes the dominant strategy in a distinct way: the $\beta$ terms control the strength of the belief signal, while the $\alpha$ reweights experts according to their beliefs.

\paragraph*{Betas} The belief differential $\beta(p^A_i) - \beta(p^B_i) = \Delta\beta\,(p^A_i - p^B_i)$, determines how strongly the expert's message reflects their beliefs. When $\Delta\beta > 0$, experts who believe $A$ is more likely to succeed ($p^A_i > p^B_i$) shift their message toward $A$, and symmetrically, experts who believe $B$ is more likely to succeed ($p^B_i > p^A_i$) shift their message toward $B$. The magnitude $\Delta\beta$ controls the strength of this incentive in both directions. 

Each expert's preferred-outcome belief $p_i^o$ depends on the sign of $\hat{v}_i = v_i + \Delta\beta\,(p^A_i - p^B_i)$, so varying $\Delta\beta$ can flip which alternative an expert prefers, thereby changing the $\alpha$ weight. For this reason, widening $\Delta\beta$ can induce additional preference flips that alter the $\alpha$-weights and may not always improve classification. Maximizing the belief signal is therefore not unconditionally optimal. In the environments we simulate, however, this effect does not bind: accuracy is monotone in $\Delta\beta$, so the optimum is attained at the budget bound $\Delta\beta = c/n - \bar{u}$ (i.e. $\beta^+ = c/n, \beta^- = \overline{u}$).


\paragraph*{Alphas} Each expert's message is scaled by $1/\alpha(p_i^o)$, where $p_i^o$ is the expert's belief that their preferred alternative will lead to a positive outcome. Since $\alpha(p)$ is linear in $p$, the relative values of $\alpha^+$ and $\alpha^-$ determine how the mechanism reweights experts based on $p_i^o$. When $\alpha^+ < \alpha^-$, experts with high $p_i^o$ have smaller $\alpha(p_i^o)$ and thus larger messages, contributing more to the aggregate. When $\alpha^+ > \alpha^-$, the opposite holds. When $\alpha^+ = \alpha^-$, all experts are weighted equally regardless of $p_i^o$.

\begin{proposition}\label{prop:alpha-ratio}
For any $\lambda > 0$, the mechanisms with parameters $(\alpha^+, \alpha^-)$ and $(\lambda\alpha^+, \lambda\alpha^-)$ select the same alternative for every type profile.
\end{proposition}

\begin{proof}
Under the dominant strategy, $m_i^* = \hat{v}_i / \alpha(p_i^o)$. Replacing $(\alpha^+, \alpha^-)$ by $(\lambda\alpha^+, \lambda\alpha^-)$:
\[
m_i^*\big|_{(\lambda\alpha^+, \lambda\alpha^-)} = \frac{\hat{v}_i}{\lambda\alpha^+ p_i^o + \lambda\alpha^-(1 - p_i^o)} = \frac{1}{\lambda\alpha^-} \cdot \frac{\hat{v}_i}{1 + p_i^o(\frac{\alpha^+}{\alpha^-} - 1) }
\]
Since $\lambda, \alpha^- > 0$, the sign of $\sum_i m_i^*$ is preserved. Hence the decision depends on $(\alpha^+, \alpha^-)$ only through the ratio $\alpha^+/\alpha^-$.
\end{proof}

When $\alpha^- > 0$, the decision can be parameterized by the ratio $\frac{\alpha^+}{\alpha^-}$. When $\frac{\alpha^+}{\alpha^-} < 1$, the denominator decreases with $p$, so experts with higher $p_i^o$ contribute more to the aggregate.
At $\alpha^- = 0$, the transfer is refunded when the outcome is negative, so the mechanism remains DSIC yet experts are never held accountable for pivotal decisions that lead to bad outcomes.
Under the dominant strategy, the expert's net financial exposure is $\alpha^- |t_i|$ and the required escrow is invariant to~$\alpha^-$. When reducing $\alpha^-$, experts compensate by submitting proportionally larger messages. The net exposure is always capped by the message the experts submit.

\paragraph*{Parameterized Designer's Objective}\label{subsec:designer-problem}
The Designer's Problem from Section~\ref{designer-objective} --- matching the sign of the aggregate message to the sign of the aggregate belief differential, as in~\eqref{sign-match} --- can now be instantiated with the derived dominant strategy.
The designer chooses the mechanism parameters $(\alpha^+, \alpha^-, \beta^+, \beta^-)$ to maximize the probability of correct classification for a given environment.
The designer's problem:
\begin{equation}\label{eq:designer-problem}
\max_{\alpha^+, \alpha^-, \beta^+, \beta^-} \; \mathbb{P}_{v_i, p^A_i, p^B_i}\!\left(\text{sgn}\!\left(\sum_{i=1}^{n} \frac{v_i + \Delta \beta (p^A_i - p^B_i)}{\alpha(p_i^{o})}\right) = \text{sgn}\!\left(\sum_{i=1}^{n} (p^A_i - p^B_i)\right)\right)
\end{equation}
subject to
\begin{align*}
& \alpha^+, \alpha^- \geq 0, \quad (\alpha^+, \alpha^-) \neq (0,0), \\
& \beta^+, \beta^- \in [\bar{u},\; c/n]
\end{align*}
for some random distributions of the agents types.

\subsection{Simulations}\label{subsec:simulation-results}

Since the belief $p_i^o$ enters the denominator of each expert's dominant strategy, and the value of $o$ depends on the sign of $\hat{v}_i$, the message aggregate is a sum of ratios of dependent random variables with no known closed form. We therefore evaluate the mechanism's performance via Monte Carlo simulation: for each parameter configuration, we repeatedly draw random type profiles from specified distributions, compute every expert's dominant-strategy message, apply the allocation rule, and report \emph{accuracy}: the fraction of the drawn profiles on which the mechanism's decision coincides with the objectively superior alternative — the one with the higher true success probability. 

A direct empirical comparison with decision markets is not well-posed in our setting. A decision market's realized accuracy is not a function of the type distribution alone; it depends on the market microstructure (trader's budget, risk attitude, the market scoring rule, and the number and order of trading rounds), none of which our mechanism requires. Moreover, under a deterministic decision rule no decision market admits a dominant strategy in the multi-agent setting; restoring strict properness requires a randomized decision rule~\cite{chen2011decision}. Quadratic voting and the standard pivotal mechanism are also excluded, because neither is able to extract beliefs from experts.

We therefore compare our mechanism against two decision rules: \emph{majority voting with rewards} (MV w/ rewards), where each expert votes for $A$ if $\hat{v}_i > 0$ and the majority wins; and \emph{majority voting without rewards} (MV), where each expert votes for $A$ if $v_i > 0$. All three rules are evaluated on the same drawn type profiles. 

We further compare against a decision scoring rule (DSR). A DSR pays a single expert a positive affine function of the principal's realized utility \cite{decision-scoring-rules}; with a binary outcome and principal utility $u = 1$ for a good outcome and $u = -1$ for a bad one, this rewards the expert for the realized impact of the recommended action. Since DSR requires a single expert, we implement the DSR as a \emph{random dictator}: one expert is drawn uniformly at random and their recommended alternative is taken.


\paragraph*{Setup}
We specify a parametric model for the experts' types, defining distributions for idiosyncratic preferences and belief noise.
The true success probabilities are symmetric around $0.5$, separated by a \emph{quality gap} parameter $\delta \in \{0.05, 0.15, 0.25\}$, corresponding to hard, moderate, and easy classification tasks, capturing the case the designer typically cares about most: one alternative is likely to succeed while the other is likely to fail. When both alternatives would succeed ($p^A, p^B$ both high) or both would fail ($p^A, p^B$ both low), the stakes of choosing correctly are lower. By the Symmetry property, we can without loss of generality let $A$ be the superior alternative.

Idiosyncratic preferences are drawn independently with mean~$\mu_v \in \{0, -1\}$, where $\mu_v = 0$ is the unbiased case and $\mu_v = -1$ represents a committee that on average opposes the superior alternative. The variance is normalized to one, so $\Delta\beta$ is measured in units of preference intensity. Each expert receives a noisy private signal centered at the truth with standard deviation~$\sigma \in \{0.2, 0.8, 2.0\}$ (\emph{belief noise}), truncated to $(0,1)$. A small $\sigma$ corresponds to a well-informed committee, while a large $\sigma$ represents dispersed information. Each expert's type is drawn as:
\[
v_i \sim \mathcal{N}(\mu_v, 1), \qquad p^A_i \sim \text{TN}\!\left(\tfrac{1}{2} + \delta,\; \sigma^2,\; 0,\; 1\right), \qquad p^B_i \sim \text{TN}\!\left(\tfrac{1}{2} - \delta,\; \sigma^2,\; 0,\; 1\right),
\]
where $\text{TN}(\mu, \sigma^2, a, b)$ denotes a truncated normal distribution on $(0, 1)$. Results with uniform preferences $v_i \sim U(\mu_v - 1, \mu_v + 1)$ are similar and omitted for brevity.

A grid search over $\alpha^+/\alpha^- \in [0, 2]$ reveals that the optimal ratio is consistently close to zero: the accuracy gap between $\alpha^+ = 0$ and the optimal $\alpha^+/\alpha^-$ never exceeds~$3\%$ across all tested parameter configurations, so we adopt $\alpha^+ = 0$ as a default.

For each parameter configuration, we draw $10^6$ independent type profiles. We vary $\Delta\beta \in [0, 5]$ on a grid of 30 equally spaced values and committee sizes $n \in \{5, 49\}$ (chosen odd to avoid ties in majority voting).

All simulations were implemented in Python with a fixed random seed. Table~\ref{tab:sim-params} summarizes the full set of varied parameters and their ranges.

\begin{table}[htbp]
\centering
\begin{tabular}{lll}
\toprule
\textbf{Parameter} & \textbf{Symbol} & \textbf{Values} \\
\midrule
Quality gap        & $\delta$        & $\{0.05,\, 0.15,\, 0.25\}$ \\
Preference bias    & $\mu_v$         & $\{0,\, -1\}$ \\
Belief noise       & $\sigma$        & $\{0.2,\, 0.8,\, 2.0\}$ \\
Committee size     & $n$             & $\{5,\, 49\}$ \\
Belief budget      & $\Delta\beta$   & $[0,\,5]$\\
\bottomrule
\end{tabular}
\caption{Parameters varied in the Monte Carlo simulations.}
\label{tab:sim-params}
\end{table}

\paragraph*{Majority Voting}
Figures~\ref{fig:accuracy-unbiased} and~\ref{fig:accuracy-biased} show that our
mechanism outperforms majority voting across almost every parameter
configuration, the exception being the small quality gap ($\delta = 0.05$) in the
large committee ($n = 49$), where the two are comparable. The margin is widest
when the belief budget $\Delta\beta$ is small.

\begin{figure}[h]
\centering
\includegraphics[width=0.69\textwidth, trim=0cm 0cm 0cm 0cm, clip]{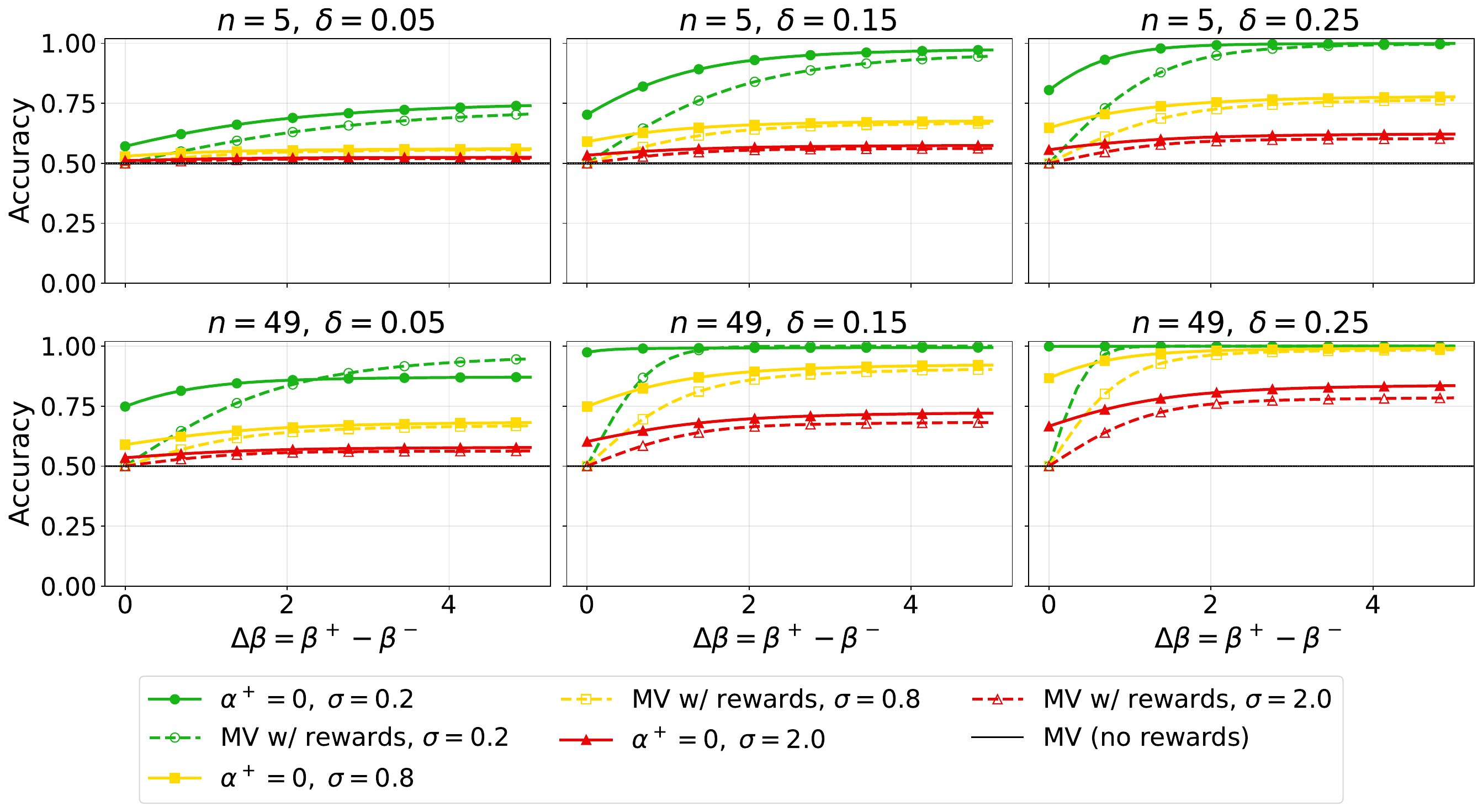}
\caption{Unbiased committee ($\mu_v = 0$): classification accuracy as a function of $\Delta\beta$. Solid: our mechanism ($\alpha^+ = 0$). Dashed: MV with rewards. MV without rewards is the horizontal line at 0.5. Rows: committee size $n$. Columns: quality gap~$\delta$. Colors: belief noise~$\sigma$. Preferences: $v_i \sim \mathcal{N}(\mu_v, 1)$.}
\label{fig:accuracy-unbiased}
\end{figure}

At $\Delta\beta = 0$ the mechanism already reweighs experts through $1/\alpha(p_i^o)$, so beliefs scale each report even with no belief budget. For $\Delta\beta > 0$, messages are additionally pushed toward the expert's belief differential and accuracy rises steeply before flattening, as the belief signal comes to dominate preference noise. The quality gap $\delta$ and belief noise $\sigma$ act as expected: a larger $\delta$ separates the true success probabilities and eases the task for every method, while a larger $\sigma$ disperses the signals and degrades accuracy.

From a design standpoint, $\Delta\beta$ is bounded above by $c/n - \bar{u}$, so the designer can pick the minimum $\Delta\beta$ -- and hence the minimum per-expert budget the DAO must commit -- needed to reach a target accuracy.

\begin{figure}[h]
\centering
\includegraphics[width=0.69\textwidth, trim=0cm 0cm 0cm 0cm, clip]{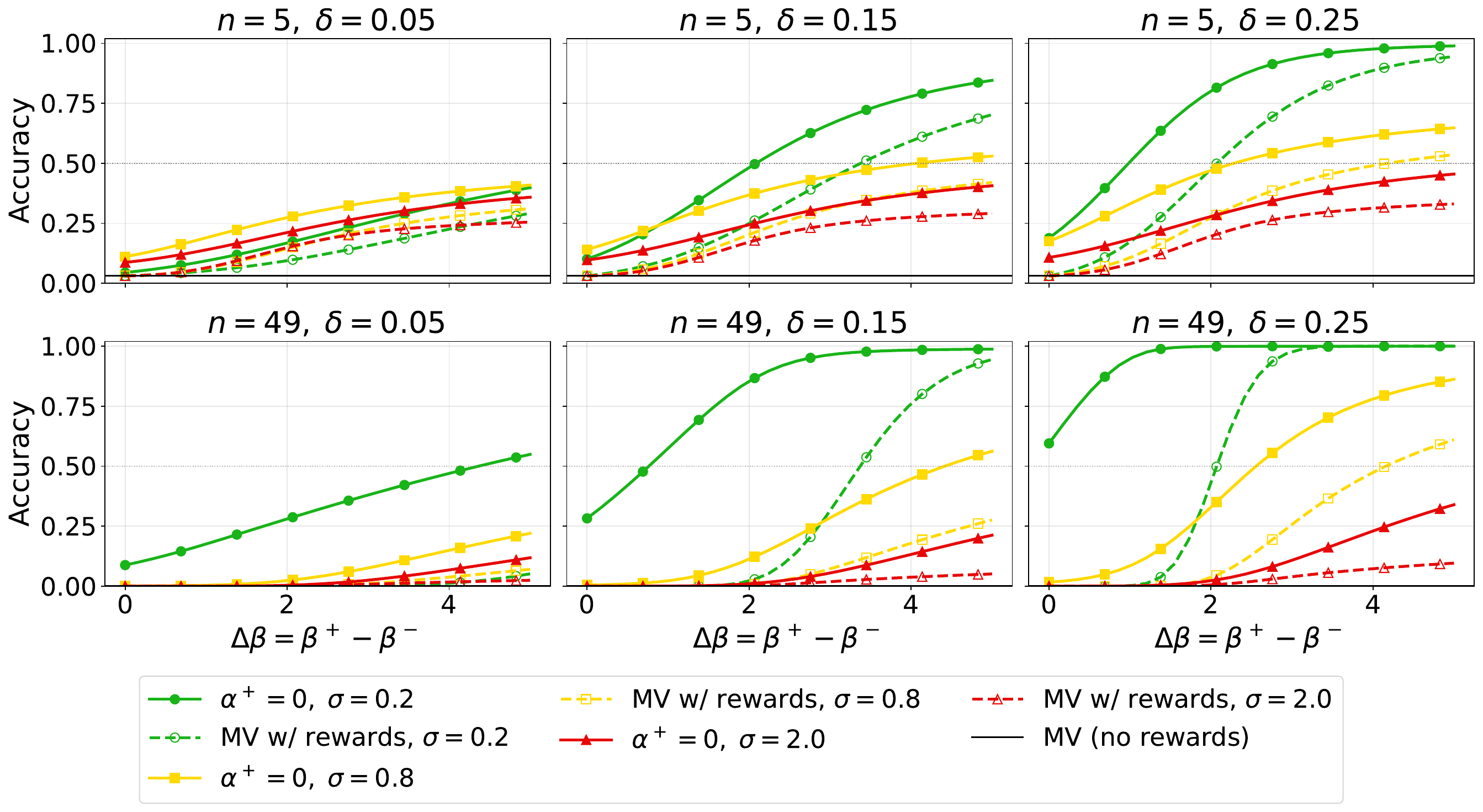}
\caption{Biased committee ($\mu_v = -1$), where members on average prefer the inferior alternative. Conventions as in Figure~\ref{fig:accuracy-unbiased}, except for the MV with no rewards that has a value close to 0.}
\label{fig:accuracy-biased}
\end{figure}

For the biased committee (Figure~\ref{fig:accuracy-biased}), majority voting without rewards has approximately 0 accuracy, and with rewards it keeps only the sign of the shifted valuation, discarding the magnitude. Our mechanism instead scales each report by $1/\alpha(p_i^o)$, so an expert backing an option they believe is unlikely to succeed carries little weight in the aggregate.

Results with uniform preferences $v_i \sim U(\mu_v - 1, \mu_v + 1)$ are qualitatively similar, with one exception. At $\mu_v = -1$ the support is $[-2, 0]$, so every expert prefers~$B$: at $\Delta\beta = 0$ no expert ever votes for~$A$, reweighting alone cannot shift the decision, and accuracy is zero -- unlike the normal case, whose tails leave some experts with $v_i > 0$. Once $\Delta\beta > 0$ the belief signal flips some experts toward~$A$ and the mechanism recovers.

\paragraph*{Random dictator}
Figure~\ref{fig:dictator-unbiased} shows that in the unbiased committee our mechanism outperforms the DSR across every configuration.

\begin{figure}[h!]
\centering
\includegraphics[width=0.69\textwidth]{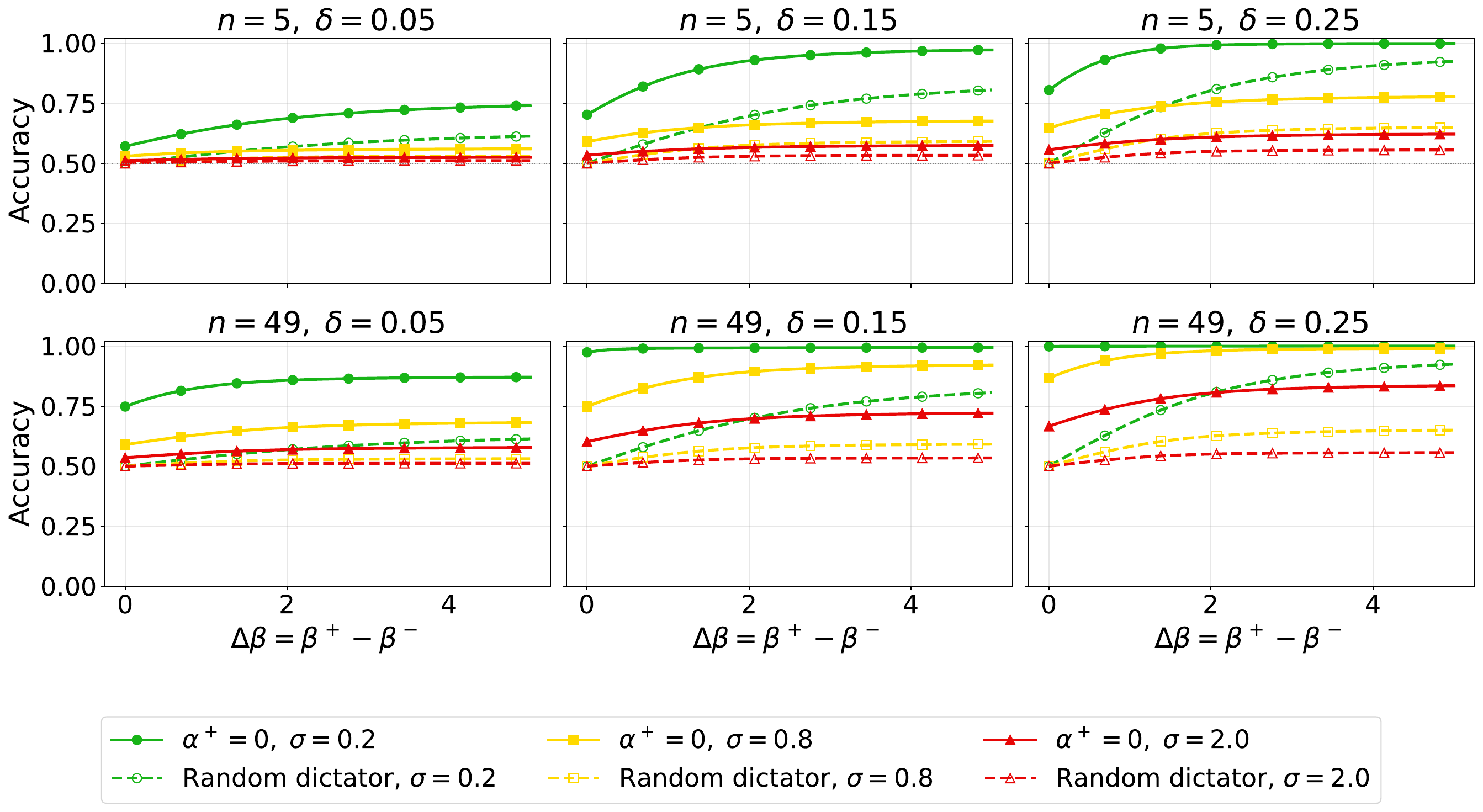}
\caption{Unbiased committee ($\mu_v = 0$): classification accuracy of our
mechanism ($\alpha^+ = 0$, solid) versus the random dictator with the same
outcome-contingent incentives (dashed), as a function of $\Delta\beta$. Conventions as in Figure~\ref{fig:accuracy-unbiased}
$\sigma$. Preferences: $v_i \sim \mathcal{N}(\mu_v, 1)$.}
\label{fig:dictator-unbiased}
\end{figure}

In the biased committee (Figure~\ref{fig:dictator-biased}) the DSR is more accurate at small $\Delta\beta$ in some scenarios.

\begin{figure}[h!]
\centering
\includegraphics[width=0.69\textwidth]{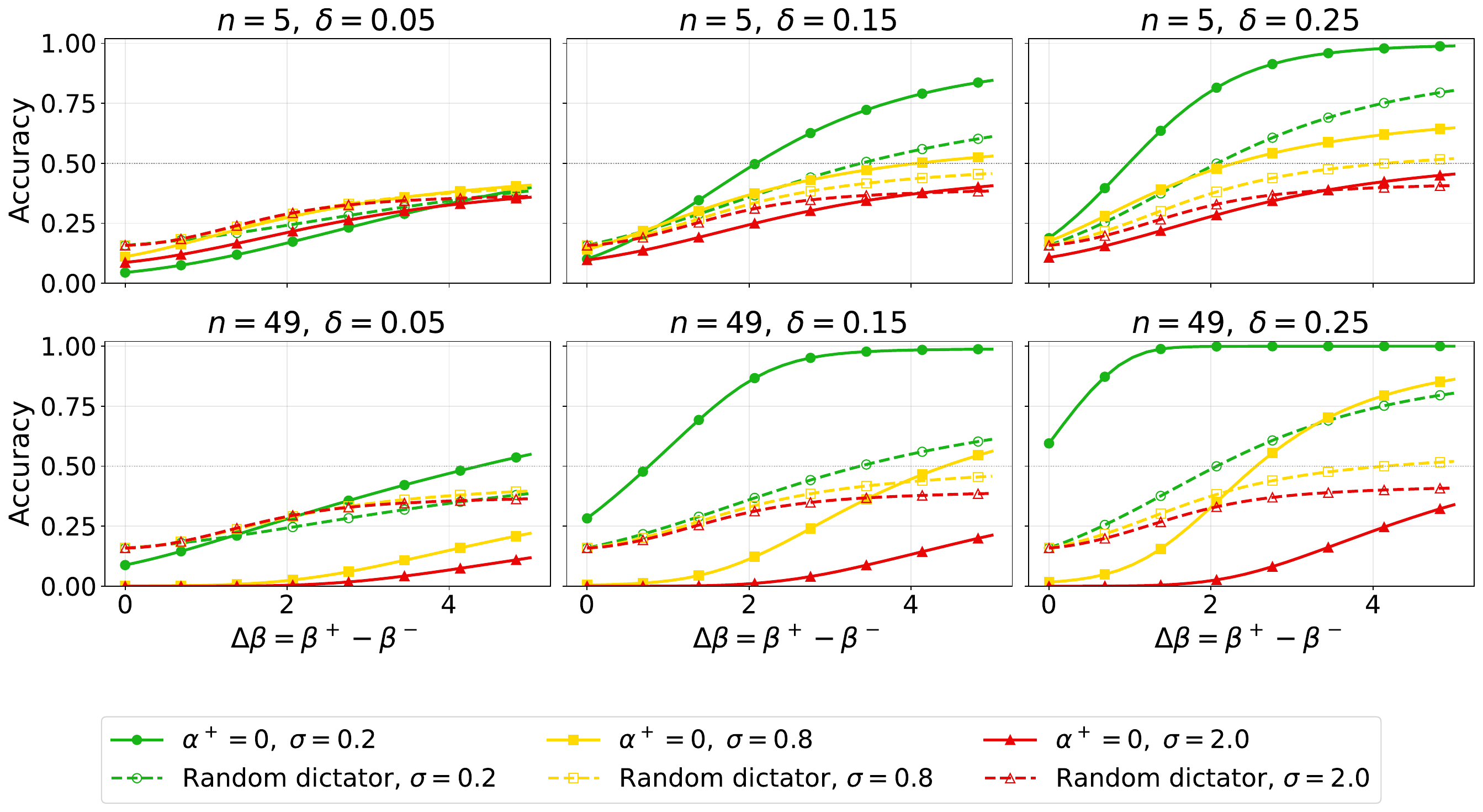}
\caption{Biased committee ($\mu_v = -1$): our mechanism ($\alpha^+ = 0$, solid)
versus the random dictator (dashed). Conventions as in Figure~\ref{fig:accuracy-unbiased}.}
\label{fig:dictator-biased}
\end{figure}


\section{Conclusions}\label{sec:conclusions}

This paper presents a mechanism for committee decision-making with outcome-contingent transfers that satisfies dominant-strategy incentive compatibility, interim individual rationality, and budget feasibility. Each expert submits a single scalar that jointly encodes preferences and beliefs, and the mechanism maximizes the probability that its decision agrees with a linear opinion pool of the experts' private beliefs. The framework abstracts away the medium of exchange, the method of outcome evaluation, and the committee election process, providing a reference construction that practical designs can instantiate. The mechanism provides value even without a budget: the $\alpha$ reweighs experts, so informative beliefs can enter the aggregate, yielding higher classification accuracy than a standard pivotal mechanism at no additional cost. The budget amplifies this effect by adding an explicit belief shift to each expert's message.

When experts have no idiosyncratic preferences, the dominant strategy reduces to a scaled belief differential and the mechanism selects the alternative according to the linear opinion pool. This elicitation is strictly incentive compatible whenever each expert assigns non-zero probability to every possible aggregate of the other experts' reports. Decision scoring rules cannot achieve this in the multi-agent setting: under a deterministic decision rule, each expert's expected payoff depends on the collective decision, which in turn depends on the other experts' reports, so agents have an incentive to misreport to change the collective decision in their favor~\cite{decision-markets}. Restoring properness requires randomized decision rules~\cite{chen2011decision}, which introduce practical difficulties. Our mechanism sidesteps this impossibility by inducing a dominant strategy under a deterministic allocation rule and across multiple agents.

We focus on binary choices, which cover the majority of DAO proposals~\cite{UnderstandingDAOs}. The designer's objective reduces to matching the sign of a linear opinion pool. Although the linear opinion pool generalizes naturally to $k$ alternatives, the mechanism design side is less straightforward. With two alternatives, the classification boundary is a single threshold that maps to the pivotal mechanism's allocation rule. With $k \geq 3$, the designer must match an $\arg\max$ over $k$ pooled probabilities, and the interaction between outcome-contingent rewards and a multi-dimensional message space remains an open problem.






\bibliography{camara-ready}

@article{UnderstandingDAOs,
author = {Wang, Qin and Yu, Guangsheng and Sai, Yilin and Sun, Caijun and Nguyen, Lam and Chen, Shiping},
year = {2025},
month = {01},
pages = {1-19},
title = {Understanding {DAOs}: An Empirical Study on Governance Dynamics},
volume = {PP},
journal = {IEEE Transactions on Computational Social Systems},
doi = {10.1109/TCSS.2025.3539889}
}

@article{cong2025centralized,
  title={Centralized governance in decentralized organizations},
  author={Cong, Lin William and Rabetti, Daniel and Wang, Charles CY and Yan, Yu},
  journal={Available at SSRN 5168660},
  year={2025}
}

@article{rikken2019governance,
  title={Governance challenges of blockchain and decentralized autonomous organizations},
  author={Rikken, Olivier and Janssen, Marijn and Kwee, Zenlin},
  journal={Information Polity},
  volume={24},
  number={4},
  pages={397--417},
  year={2019},
  publisher={SAGE Publications Sage UK: London, England}
}

@inproceedings{kitzler2024governance,
  title={The Governance of Decentralized Autonomous Organizations: A Study of Contributors' Influence, Networks, and Shifts in Voting Power},
  author={Kitzler, Stefan and Balietti, Stefano and Saggese, Pietro and Bernhard, Hans-Martin and Strohmaier, Markus},
  booktitle={Financial Cryptography and Data Security (FC)},
  year={2024}
}

@article{han2025dao,
  title={{DAO} Governance},
  author={Han, Jiasun and Lee, Jongsub and Li, Tao},
  journal={Journal of Corporate Finance},
  year={2025},
  publisher={Elsevier}
}

@article{dimitri2023voting,
  title={Voting in {DAOs}},
  author={Dimitri, Nicola},
  journal={Distributed ledger technologies: Research and practice},
  volume={2},
  number={4},
  pages={1--12},
  year={2023},
  publisher={ACM New York, NY}
}

@inproceedings{talmon2023social,
  title={Social Choice Around Decentralized Autonomous Organizations: On the Computational Social Choice of Digital Communities},
  author={Talmon, Nimrod},
  booktitle={Proceedings of the 22nd International Conference on Autonomous Agents and Multiagent Systems (AAMAS)},
  pages={1768--1773},
  year={2023}
}

@misc{sharma2024future,
  title={Future of algorithmic organization: Large-scale analysis of decentralized autonomous organizations ({DAOs}). {arXiv}},
  author={Sharma, T and Potter, Y and Pongmala, K and Wang, H and Miller, A and Song, D and Wang, Y},
  year={2024}
}

@article{davo2025rise,
  title={The rise and fall of {DAOstack}: lessons for decentralized autonomous organizations},
  author={Dav{\'o}, David and Arroyo, Javier and Hassan, Samer and Semenzin, Silvia},
  journal={PeerJ Computer Science},
  volume={11},
  pages={e3320},
  year={2025},
  publisher={PeerJ Inc.}
}

@misc{daostack_whitepaper,
  title = {{DAOstack} Whitepaper: An Operating System for Collective Intelligence},
  author = {{DAOstack}},
  year = {2018},
  howpublished = {\url{https://www.allcryptowhitepapers.com/daostack-whitepaper/}},
  note = {Accessed: 2026-03-14}
}

@misc{metadao,
  title = {{MetaDAO}},
  author = {{MetaDAO}},
  year = {n.d.},
  url = {https://www.metadao.fi/},
  note = {Accessed: 2026-03-13}
}

@inproceedings{kiayias2022sok,
  title={{SoK}: blockchain governance},
  author={Kiayias, Aggelos and Lazos, Philip},
  booktitle={Proceedings of the 4th ACM Conference on Advances in Financial Technologies},
  pages={61--73},
  year={2022}
}

@misc{bip2,
  author       = {Luke Dashjr},
  title        = {{BIP} 2: {BIP} process, revised},
  year         = {2016},
  howpublished = {\url{https://github.com/bitcoin/bips/blob/master/bip-0002.mediawiki}},
  note         = {Bitcoin Improvement Proposal}
}

@article{becze2015eip,
  title={{EIP}-1: {EIP} purpose and guidelines},
  author={Becze, Martin and Jameson, Hudson and others},
  journal={Ethereum Improvement Proposals},
  volume={27},
  year={2015}
}

@misc{polkadot_web3,
  title        = {{Web3} and {Polkadot}},
  author       = {{Polkadot Wiki}},
  year         = {2025},
  url          = {https://wiki.polkadot.com/general/web3-and-polkadot/},
  note         = {Accessed: 2026-03-12}
}

@misc{projectcatalyst_docs,
  author       = {{Project Catalyst}},
  title        = {Project Catalyst Documentation},
  year         = {n.d.},
  howpublished = {\url{https://docs.projectcatalyst.io/}},
  note         = {Accessed: 2026-03-12}
}

@misc{lesaege2019kleros,
  author  = {Lesaege, Cl{\'e}ment and Ast, Federico and George, William},
  title   = {Kleros: Short Paper v1.0.7},
  year    = {2019},
  url     = {https://kleros.io/whitepaper.pdf},
  urldate = {2026-03-25}
}

@book{de2014essai,
  title={Essai sur l'application de l'analyse {\`a} la probabilit{\'e} des d{\'e}cisions rendues {\`a} la pluralit{\'e} des voix},
  author={de Condorcet, Nicolas},
  year={1785},
  publisher={Imprimerie Royale, Paris}
}

@article{austensmith1996,
  title={Information Aggregation, Rationality, and the {Condorcet} Jury Theorem},
  author={Austen-Smith, David and Banks, Jeffrey S.},
  journal={American Political Science Review},
  volume={90},
  number={1},
  pages={34--45},
  year={1996},
  publisher={Cambridge University Press}
}

@techreport{li1999conflicts,
  title={Conflicts and Common Interests in Committees},
  author={Li, Hao and Rosen, Sherwin and Suen, Wing},
  type={NBER Working Paper},
  number={7158},
  institution={National Bureau of Economic Research},
  year={1999}
}

@article{nitzan1982,
  title={Optimal Decision Rules in Uncertain Dichotomous Choice Situations},
  author={Nitzan, Shmuel and Paroush, Jacob},
  journal={International Economic Review},
  volume={23},
  number={2},
  pages={289--297},
  year={1982},
  publisher={Wiley}
}

@inproceedings{decision-scoring-rules,
  title={Decision Scoring Rules.},
  author={Oesterheld, Caspar and Conitzer, Vincent},
  booktitle={WINE},
  pages={468},
  year={2020}
}

@inproceedings{decision-markets,
  title={Decision rules and decision markets.},
  author={Othman, Abraham and Sandholm, Tuomas},
  booktitle={AAMAS},
  pages={625--632},
  year={2010}
}

@article{boutilier2011eliciting,
  title={Eliciting forecasts from self-interested experts: scoring rules for decision makers},
  author={Boutilier, Craig},
  journal={arXiv preprint arXiv:1106.2489},
  year={2011}
}

@inproceedings{chen2011decision,
  title={Decision markets with good incentives},
  author={Chen, Yiling and Kash, Ian and Ruberry, Mike and Shnayder, Victor},
  booktitle={International Workshop on Internet and Network Economics},
  pages={72--83},
  year={2011},
  organization={Springer}
}

@article{hanson2013shall,
  title={Shall We Vote on Values, But Bet on Beliefs?},
  author={Hanson, Robin},
  journal={Journal of Political Philosophy},
  volume={21},
  number={2},
  pages={151--178},
  year={2013}
}

@article{miller2005eliciting,
  title={Eliciting informative feedback: The peer-prediction method},
  author={Miller, Nolan and Resnick, Paul and Zeckhauser, Richard},
  journal={Management Science},
  volume={51},
  number={9},
  pages={1359--1373},
  year={2005},
  publisher={INFORMS}
}

@article{winkler2019probability,
  title={Probability Forecasts and Their Combination: A Research Perspective},
  author={Winkler, Robert L. and Grushka-Cockayne, Yael and Lichtendahl, Kenneth C. and Jose, Victor Richmond R.},
  journal={Decision Analysis},
  volume={16},
  number={4},
  pages={239--260},
  year={2019},
  publisher={INFORMS}
}

@article{ranjan2010combining,
  title={Combining probability forecasts},
  author={Ranjan, Roopesh and Gneiting, Tilmann},
  journal={Journal of the Royal Statistical Society Series B: Statistical Methodology},
  volume={72},
  number={1},
  pages={71--91},
  year={2010},
  publisher={Oxford University Press}
}

@book{green1979incentives,
  title={Incentives in public decision-making},
  author={Green, Jerry and Laffont, Jean-Jacques},
  year={1979},
  publisher={Elsevier North-Holland}
}

@article{myerson1981optimal,
  title={Optimal Auction Design},
  author={Myerson, Roger B.},
  journal={Mathematics of Operations Research},
  volume={6},
  number={1},
  pages={58--73},
  year={1981},
  publisher={INFORMS}
}

@article{mezzetti2004mechanism,
  title={Mechanism design with interdependent valuations: Efficiency},
  author={Mezzetti, Claudio},
  journal={Econometrica},
  volume={72},
  number={5},
  pages={1617--1626},
  year={2004},
  publisher={Wiley Online Library}
}

@article{nath2013strict,
  title={A strict ex-post incentive compatible mechanism for interdependent valuations},
  author={Nath, Swaprava and Zoeter, Onno},
  journal={Economics Letters},
  volume={121},
  number={2},
  pages={321--325},
  year={2013},
  publisher={Elsevier}
}

@Inproceedings{lalley2018quadratic,
  title={Quadratic voting: How mechanism design can radicalize democracy},
  author={Lalley, Steven P and Weyl, E Glen},
  booktitle={AEA Papers and Proceedings},
  volume={108},
  pages={33--37},
  year={2018},
  organization={American Economic Association 2014 Broadway, Suite 305, Nashville, TN 37203}
}

@article{Flexible-Design-for-Public-Goods,
author = {Buterin, Vitalik and Hitzig, Zo\"{e} and Weyl, E. Glen},
title = {A Flexible Design for Funding Public Goods},
year = {2019},
issue_date = {November 2019},
publisher = {INFORMS},
address = {Linthicum, MD, USA},
volume = {65},
number = {11},
issn = {0025-1909},
url = {https://doi.org/10.1287/mnsc.2019.3337},
doi = {10.1287/mnsc.2019.3337},
journal = {Manage. Sci.},
month = nov,
pages = {5171–5187},
numpages = {17}
}

@article{kagel1986winner,
  title={The winner's curse and public information in common value auctions},
  author={Kagel, John H and Levin, Dan},
  journal={The American economic review},
  pages={894--920},
  year={1986},
  publisher={JSTOR}
}

@article{gibbard1973manipulation,
  title={Manipulation of Voting Schemes: A General Result},
  author={Gibbard, Allan},
  journal={Econometrica},
  volume={41},
  number={4},
  pages={587--601},
  year={1973},
  publisher={The Econometric Society}
}


\end{document}